\documentclass{amsart}
\usepackage[a4paper,margin=2.5cm]{geometry}
 
\usepackage[T1]{fontenc}
\usepackage{amssymb, amsmath, amsthm, mathtools, stmaryrd}
\usepackage{enumitem}

\usepackage{url}

\usepackage[backend=bibtex,style=alphabetic,sorting=nyt,isbn=false,url=false,doi=true,maxalphanames=10,minalphanames=4,mincitenames=4,maxcitenames=10,minnames=4,maxnames=10,giveninits=true,maxbibnames=99]{biblatex}

\usepackage{xcolor}
\usepackage[colorlinks=true]{hyperref}

\newcommand{\Res}{\mathop{\mathrm{Res}}}

\newcommand{\sumk}{\mathbf{k}}

\newcommand{\coeff}[1]{\mathop{\mathrm{Coeff}_{[{#1}]}}}
\newcommand{\restr}[2]{\mathop{\big\lfloor_{{#1}\to {#2}}}}
\newcommand{\set}[1]{\llbracket {#1} \rrbracket}

\usepackage{mathtools}

\numberwithin{equation}{section} 

\theoremstyle{plain}
\newtheorem{theorem}{Theorem}[section]
\newtheorem{conjecture}[theorem]{Conjecture}
\newtheorem{proposition}[theorem]{Proposition}
\newtheorem{lemma}[theorem]{Lemma}

\theoremstyle{definition}

\newtheorem{remark}[theorem]{Remark}

\title[Resonance transformations: a proof of Artemev's conjecture]{Resonance transformations for the $(2,2p+1)$ minimal string via $x-y$ swap: a proof of Artemev's conjecture}

\author{Kornelis Dekinga}

\address{KdV Institute for Mathematics and Institute of Physics, Universiteit van Amsterdam, Amsterdam, Nederland}
\email{kornelis.dekinga@student.uva.nl}	

\author{Sergey Shadrin}

\address{Korteweg-de Vriesinstituut voor Wiskunde, Universiteit van Amsterdam, Postbus 94248, 1090GE Amsterdam, Nederland}
\email{s.shadrin@uva.nl}	

\author{Erik Verlinde}

\address{Institute of Physics, Universiteit van Amsterdam, Postbus 94485
	1090 GL Amsterdam, 1090GE Amsterdam, Nederland}
\email{e.p.verlinde@uva.nl}

\begin{document}

\begin{abstract} This paper contains a proof of a recent conjecture of Artemev that connected the resonance transformations for the $(2,2p+1)$ minimal string to the $x-y$ swap in the theory of topological recursion.  
\end{abstract}

\maketitle

\tableofcontents

\section{Introduction}









\subsection{\texorpdfstring{$(p,q)$}{(p,q)} minimal string theory}
In the worldsheet formalism $(p,q)$ minimal string theory is defined as a $(p,q)$ minimal model conformal field theory coupled to Liouville theory. Interesting quantities in this theory are amplitudes of tachyons. Tachyons are obtained by first dressing a conformal matter primary $\mathcal{O}_{p,q}$ by a Liouville primary $V_{\alpha_{p,q}}$ such that the combined conformal dimension is $(1,1)$, and then multiplying by the ghosts $\mathfrak{c}\overline{\mathfrak{c}}$, see e.g.~\cite{SibergShih}. The amplitudes are then computed by integrating CFT correlators of these operators over moduli space. Another approach to 2D quantum gravity is via discretized random surfaces. In this approach the worldsheets of strings are approximated by Feynman graphs in the limit where the number of vertices becomes infinite. These Feynman graphs can be described by random matrices, which reduces the problem to that of a matrix integral~\cite{GrossMigdal1990,BrezinKazakov1990,DouglasShenker1990}. To ensure that the matrix models reproduce the same results as the worldsheet formalism, the potential needs to be tuned in a specific way while taking the $N\to\infty$ limit. In~\cite{Douglas1990} a formula is proposed that gives the right tuning to obtain the $(p,q)$ model. This formula is formulated in terms of two differential operators $P$ and $Q$ of order $p$ and $q$ respectively, and the partition function can be extracted from these operators. However, it was already noted in~\cite{moore-seiberg} that there is a discrepancy between the matrix model and worldsheet sides. The proposed solution is a nonlinear transformation of the coupling constants on the matrix model side, which is now often referred to as the ``resonance transformation''. For the $(2,2p+1)$ minimal string this has been confirmed to work for the genus zero, three and four point functions in~\cite{BelavinZam-resonance}.

It is believed that higher loop corrections are governed by topological recursion. The spectral curve for $(p,q)$ minimal string theory is given by two Chebyshev polynomials of the first kind. It has been shown that this choice reproduces the expected fusion rules in the $(2,2p+1)$ case~\cite{marshakov}. See~\cite{SibergShih} for an interpretation of the spectral curve in terms of the FZZT partition function. In this approach it is also necessary to use the resonance transformation to make the results agree with the worldsheet formalism. In \cite{artemev} it was conjectured that the resonance transformation is unnecessary if one swaps $x$ and $y$ in the spectral curve. A similar idea was proposed and confirmed for the three and four point functions on the sphere in~\cite{BelavinRud}, where the connection to Frobenius manifolds was used~\cite{BelavinDubrovinMukh}.

These conjectures suggest that the resonance transformation is closely related to swapping $p$ and $q$. This idea is strengthened by the work of~\cite{Fukuma1992}, where a relation was found between the operators $P$ and $Q$ of the $(p,q)$ solution to the Douglas string equation to those of the $(q,p)$ solution. The coefficients of these solutions are related via a transformation that has the same structure as the resonance transformation. Another interpretation of the resonance transformation is given in~\cite{BelavinZam-resonance,BelavinDubrovinMukh}. It is possible to add a delta type term to the CFT correlators in the moduli integral when two points coincide. This can be written as a resonance transformation of the coupling constants.

\subsection{\texorpdfstring{$x-y$}{x-y} swap}
The technique of $x-y$ swap is a powerful tool in the framework of topological recursion of Chekhov-Eynard-Orantin~\cite{CEO,EynardOrantin-topologicalrecursion,Eynard-Lecture} that leads itself to a full revision and generalization of this theory~\cite{ABDKS-gen}. The explicit closed differential-algebraic formula was first conjectured in~\cite{borot2023functionalrelationshigherorderfree}. Subsequently, it was proved (under the assumption of the existence of the loop insertion operator) in genus $0$ in~\cite{hock2022xy}, drastically simplified in~\cite{HockFormula}, and finally proved in full generality in~\cite{alexandrov2022universal}. The $x-y$ swap has very strong ties with KP integrability~\cite{ABDKS-KPxy}, and in this way it ultimately generalizes the formulas  for the $p-q$ duality in~\cite{KharchevMarshakov}. 

The latter identification of $p-q$ duality as the special case of the $x-y$ swap might explain why the induced change of flat parametrization in~\cite[Eq.~(4.15)]{Fukuma1992} coincides (up to tuning of some regular terms) with the resonance tranformations in~\cite{Tarno,artemev}. Indeed, the formulas for the $p-q$ duality in~\cite{KharchevMarshakov} are derived from the swap of Kac-Schwarz operators that follows from the $(p,q)$ swap in the context of the Douglas equation in~\cite{Fukuma1992}. 

However, we note that a direct connection between the Douglas equation and topological recursion is not really worked out in the literature, except for the computations done in~\cite{BelavinDubrovinMukh} that connect the Douglas equation to the Dubrovin-Frobenius manifolds theory~\cite{Dub-2dtft}, which in turn can be placed in the context of topological recursion via its Landau-Ginzburg superpotential description~\cite{DNOPS-1,DNOPS-2}. Moreover, in~\cite{marshakov} it is argued that there should be some discrepancy that has to be taken into account, and indeed it was noticed in~\cite{artemev} and confirmed by the computation below that there is some mismatch of the regular terms (luckily, irrelevant for the computation of the correlators).

\subsection{Notation}

Throughout the paper we use the following notation: 
\begin{itemize}
	\item $(x)_p \coloneqq \Gamma(x+p)/\Gamma(x)=x(x+1)\cdots (x+p-1)$.
	\item $\set{m}\coloneqq \{1,\dots,m\}$.
	\item For any $I\subseteq \set{m}$, $|I|$ is the cardinality of $I$.
	\item Let $k_i$ be numbers indexed by $i\in \set{m}$. Then for any $I\subseteq \set{m}$, $\sumk_I\coloneqq \sum_{i\in I} k_i$. 
	\item Let $z_i$ be formal variables indexed by $i\in \set{m}$. Then for any $I\subseteq \set{m}$, $z_I\coloneqq \{z_i\}_{i\in I}$. 
	\item $T_d(z)$ denotes the Chebyshev polynomial of degree $d$. 
\end{itemize}

\subsection{Acknowledgements} We thank A.~Artemev for very useful discussions and important insights. 

S.~S. was supported by the Dutch Research Council, grant no. OCENW.M.21.233. 

\section{Basic setup for Artemev's conjecture}

\subsection{Minimal string spectral curve side} We follow~\cite[Sec.~3]{artemev}.
Apply the topological recursion procedure to the curve $\mathbb{C}\mathrm{P}^1$ with the standard Bergman kernel given in a global affine coordinate $z$ by $B(z_1,z_2) = \frac{dz_1dz_2}{(z_1-z_2)^2}$ and with the functions $x$ and $y$ given in terms of the Chebyshev polynomials as 
\begin{align} \label{eq:MinimalStringSpectralCurve}
	x & = 2 u T_2(z) = u\Big(w^2 + \frac 1 {w^2}\Big),
	& 
	y & = 2 u^{p+\frac 12} T_{p+1}(z) = u^{p+\frac 12} \Big(w^{2p+1} + \frac 1 {w^{2p+1}}\Big),
\end{align}
where $w$ is a global affine coordinate on the double cover of $\mathbb{C}\mathrm{P}^1$ given by $z = \frac 12 (w+\frac 1w)$. It produces a system of symmetric $n$- differentials $\omega^{(g)}_n= \omega^{(g)}_n(z_{\set{n}})$, $g\geq 0$, $n\geq 1$. In particular, 
\begin{align}
	\omega^{(0)}_1(z_1) & = y(z_1) dx(z_1),
	&\omega^{(0)}_2(z_1,z_2) &= B(z_1,z_2),
\end{align}
and for all $2g-2+n>0$ the differentials $\omega^{(g)}_n$ have poles only at the divisors $z_i = p$, $i=1,\dots,n$, for the points $p\in \mathbb{C}\mathrm{P}^1$ such that $dx|_p = 0$.

The main object of study is a system of Laurent monomials in $u$ denoted by $A^{(g)}_n(k_1,\dots,k_n)$, $g\geq 0$, $n\geq 1$, $k_1,\dots,k_n=1,\dots,p$ and defined as
\begin{align} \label{eq:A-definiton}
	A^{(g)}_n(k_1,\dots,k_n) & \coloneqq u^{-\sum_{i=1}^n (k_i+1)} \frac{\partial^n F^{(g)}}{\prod_{i=1}^n\partial \tau_{k_i}} \bigg|_{\tau_1=-\frac 12, \tau_2=\cdots=\tau_p=0},
\end{align}
where $F^{(g)}$ is given in terms of the expansions of $\omega^{(g)}_n$ at $z=0$ (or, more conveniently, at $w=0$) as 
\begin{align} \label{eq:DefFg}
	F^{(g)}(t_1,\dots,t_p) 
	& = \sum_{n=1}^\infty \frac{(-2)^n}{n!} \sum_{k_1,\dots,k_n=1}^p \prod_{i=1}^n \Bigg( (t_{k_i}-t_{k_i}^o) \Res_{w_i=0}\frac{x(w_i)^{p-k_i+\frac 12}}{2(p-k_i+\frac 12)} \Bigg) \omega^{(g)}_n
	\\ \notag 
	& = \sum_{n=1}^\infty \frac{(-1)^n}{n!} \sum_{k_1,\dots,k_n=1}^p \prod_{i=1}^n \Bigg( (t_{k_i}-t_{k_i}^o) \Res_{w_i=0}\frac{\Big(u\big(w_i^2+\frac 1{w_i^2}\big)\Big)^{p-k_i+\frac 12}}{(p-k_i+\frac 12)} \Bigg) \omega^{(g)}_n
\end{align}
(the factor of $(-2)^n$ is missed in~\cite{artemev}, but it ought to be there, cf.~\cite[Eq.~(3.5)]{artemev}).

The relation between the variables $\{t_k\}$ and $\{\tau_k\}$, where $k=1,2,\dots,p$, is given by
\begin{align} \label{eq:t-tau-relation}
		t_l & = (2p+1) u^{l+1} \sum_{n=1}^\infty \sum_{\substack {m_1,\dots,m_n\geq 1 \\ \sum_{i=1}^n(m_i+1) = l+1}}\frac{ (2p-2l+2n-3)!!}{n! (2p-2l-1)!!}  \prod_{i=1}^n \tau_{m_i} 
		\\ \notag &
		=(2p+1) u^{l+1} \sum_{n=1}^\infty \frac{(-2)^{n-1}}{n} \binom{-p+l-\frac 12}{n-1} \sum_{\substack {m_1,\dots,m_n\geq 1 \\ \sum_{i=1}^n(m_i+1) = l+1}} \prod_{i=1}^n \tau_{m_i} 
\end{align}
and for the reference point of expansion in $t$-variables we have
\begin{align}
	t_k^o & = -\frac 12 \Res_{w=0} x(w)^{-p+k-\frac 12} y(w)dx(w)
	\\ \notag 
	& = -\frac 12 \Res_{w=0} \bigg(u\Big(w^2+\frac 1{w^2}\Big)\bigg)^{-p+k-\frac 12} \bigg(u^{p+\frac 12}\Big(w^{2p+1}+\frac 1{w^{2p+1}}\Big)\bigg) d\bigg(u\Big(w^2+\frac 1{w^2}\Big)\bigg)
	\\ \notag &
	= {u^{k+1}} \Res_{w=0} \frac{(1+w^{4p+2})(1-w^4)}{(1+w^4)^{p-k+\frac 12}} \frac{dw}{w^{2k+3}}
	= {u^{k+1}} \Res_{w=0} \frac{(1-w^4)}{(1+w^4)^{p-k+\frac 12}} \frac{dw}{w^{2k+3}}
			\\ \notag &
	= {u^{k+1}}\Bigg( \binom{-p+k-\frac 12}{\frac{k+1}2} - \binom{-p+k-\frac 12}{\frac{k-1}2}\Bigg) \qquad \qquad k=1,3,5,\dots 
	\\ \notag &
	= (-1)^{\frac{k+1}{2}}u^{k+1}(2p+1) \frac{(2p-k-2)!!}{(k+1)!!(2p-2k-1)!!}  \qquad \qquad k=1,3,5,\dots,
\end{align}
and $0$ for $k=2,4,6,\dots$, which corresponds to $\tau^o_1=-\frac 12$, $\tau^o_2=\cdots=\tau^o_p=0$ in the $\tau$-variables. Indeed,
\begin{align}
& (2p+1) u^{k+1} \sum_{n=1}^\infty \sum_{\substack {m_1,\dots,m_n\geq 1 \\ \sum_{i=1}^n(m_i+1) = k+1}}\frac{ (2p-2k+2n-3)!!}{n! (2p-2k-1)!!}  \prod_{i=1}^n \tau_{m_i} \bigg|_{\tau_1=-\frac 12, \tau_2=\cdots=\tau_p=0}
\end{align}
is equal to zero for even $k$ and to the expression above for $k=1,3,5,\dots$. 

Using that $\omega^{(g)}_n$ is a monomial in $u$ of degree $-(2g-2+n)(p+\frac 32)$, it is straightforward to check that $A^{(g)}_n(k_1,\dots,k_n)$ is a monomial in $u$ of degree
\begin{align} \label{eq:degree}
	(2-2g)p -(3g-3+n) - \sumk_{\set{n}}.
\end{align}
Let $A^{(g),\mathsf{sing}}_n$ be a singular in $u^2$ part of $A^{(g)}_n$. That is,
\begin{align}
	A^{(g),\mathsf{sing}}_n(k_1,\dots,k_n) \coloneqq \begin{cases}
		0 & \big((2-2g)p -(3g-3+n) - \sumk_{\set{n}}\big) \in 2 \mathbb{Z}_{\geq 0}; \\
		A^{(g)}_n(k_1,\dots,k_n) & \mathrm{otherwise}.
	\end{cases}
\end{align} 


\subsection{The dual side and the conjecture} Apply the $x-y$ swap to the minimal sting spectral curve~\eqref{eq:MinimalStringSpectralCurve}. This means that we run the topological recursion procedure for the curve $\mathbb{C}\mathrm{P}^1$ with the standard Bergman kernel given in a global affine coordinate $z$ by $B(z_1,z_2) = \frac{dz_1dz_2}{(z_1-z_2)^2}$ and with the functions $x^\vee$ and $y^\vee$ given in terms of the Chebyshev polynomials as 
\begin{align} \label{eq:MinimalStringSpectralCurve-dual} 
	x^\vee & = 2 u^{p+\frac 12} T_{p+1}(z) = u^{p+\frac 12} \Big(w^{2p+1} + \frac 1 {w^{2p+1}}\Big),
	& 
	y^\vee & = 2 u T_2(z) = u\Big(w^2 + \frac 1 {w^2}\Big). 
\end{align}
Here $w$ is still a global affine coordinate on the double cover of $\mathbb{C}\mathrm{P}^1$ given by $z = \frac 12 (w+\frac 1w)$.
Denote the resulting differentals by $\omega^{(g),\vee}_n$, $g\geq 0$, $n\geq 1$. In particular, 
\begin{align}
	\omega^{(0),\vee}_1(z_1) & = y^\vee(z_1) dx^\vee(z_1) = x(z_1)dy(z_1),
	&\omega^{(0),\vee}_2(z_1,z_2) &= B(z_1,z_2) = \omega^{(0)}_2(z_1,z_2).
\end{align}

For any $k_1,\dots,k_n=1,\dots,p$ define $A^{(g),\vee}_n(k_1,\dots,k_n)$ as
\begin{align} \label{eq:defA-vee}
	A^{(g),\vee}_n(k_1,\dots,k_n) \coloneqq \prod_{i=1}^n \Bigg( \Res_{w_i=0} \frac{\big(x^\vee(w_i)\big)^{\frac{2p-2k_i+1}{2p+1}}}{2p-2k_i+1} \Bigg) \omega^{(g),\vee}_n .
\end{align}

\begin{conjecture}[{\cite[Sec.~4.1 and 4.2]{artemev}}] \label{conj:art}
	Up to a suitable normalization, $A^{(g),\vee}_n(k_1,\dots,k_n)$ coincides with $A^{(g),\mathsf{sing}}_n(k_1,\dots,k_n)$, $g\geq 0$, $n\geq 1$, $k_1,\dots,k_n=1,\dots,p$. 
\end{conjecture}

Note that $A^{(g),\vee}_n(k_1,\dots,k_n)$ is a monomial in $u$ of the same degree~\eqref{eq:degree} as $A^{(g)}_n(k_1,\dots,k_n)$, so some selection of a singular part of  $A^{(g),\vee}_n$ in $u^2$ might be needed. To this end, we can make a very precise statement, but it needs a slight revision of the minimal string spectral curve side. 

\subsection{Minimal string spectral curve side revisited}  

In the construction of $A^{(g),\mathsf{sing}}_n(k_1,\dots,k_n)$ the residues in the definition of $F^{(g)}$ are taken at $w_i=0$, which are the regular points of $\omega^{(g)}_n$ for $2g-2+n> 0$. However, this definition has to be clarified for $(g,n) = (0,1)$ and $(0,2)$. In the latter two cases the following formulas were used in~\cite{artemev}:
\begin{align} \label{eq:01-case}
	\frac{\partial F^{(0)}}{\partial t_k}\Big|_{t_l = t_l^o,\ l=1,\dots,p} & \coloneqq -\Res_{w=0} 
	\frac{x(w)^{p-k+\frac 12}}{(p-k+\frac 12)} 
	\omega^{(0)}_1 (w)
\end{align}
and 
\begin{align}\label{eq:02-case}
	\frac{\partial^2 F^{(0)}}{\partial t_{k_1} \partial t_{k_2}}\Big|_{t_l = t_l^o,\ l=1,\dots,p} & \coloneqq \Res_{w_1=0} \Res_{w_2=0}
	\frac{x(w_1)^{p-k_1+\frac 12}}{2(p-k_1+\frac 12)} \frac{x(w_2)^{p-k_2+\frac 12}}{2(p-k_2+\frac 12)} \omega^{(0)}_2 (w_1,w_2)
	\\ \notag & 
	= \Res_{z=\infty} \frac{x(z)^{p-k_1+\frac 12}}{(p-k_1+\frac 12)} \frac{\big(dx(z)^{p-k_2+\frac 12}\big)_+}{(p-k_2+\frac 12)},
\end{align}
where $(\alpha)_+$ means that we can the purely singular part of the $\alpha$ at $z\to \infty$ (cf.~e.g.~\cite[Eq.~(9)]{marshakov}).

We remark that the main contributon of these unstable cases consists of regular terms in $u^2$ that has to be ignored after all when we pass from $A^{(g)}_n$ to $A^{(g),\mathsf{sing}}_n$. More precisely, we have the following proposition:

\begin{proposition} \label{prop:sing-stable} Let $F^{(g),\mathsf{st}}$ (here $\mathsf{st}$ stands for \emph{stabilized}) be defined by~\eqref{eq:DefFg} with $\omega^{(0)}_1$ replaced by $0$ and $\omega^{(0)}_2$ replaced by $B(z_1,z_2)-B(w_1,w_2)$. Then, let $A^{(g),\mathsf{st}}_n(k_1,\dots,k_n)$ be defined by~\eqref{eq:A-definiton} with $F^{(g)}$ replaced by $F^{(g),\mathsf{st}}$. Then we have the following alternative formula for $A^{(g),\mathsf{sing}}_n(k_1,\dots,k_n)$ for $n\geq 2$: 
\begin{align}
	A^{(g),\mathsf{sing}}_n(k_1,\dots,k_n) \coloneqq \begin{cases}
		0 & \big((2-2g)p -(3g-3+n) - \sumk_{\set{n}} \big) \in 2 \mathbb{Z}_{\geq 0}; \\
		A^{(g),\mathsf{st}}_n(k_1,\dots,k_n)& \mathrm{otherwise}.
	\end{cases}
\end{align} 
In the case $n=1$ there is an exceptional term that is now omitted by stabilization, that is, $A^{(0),\mathsf{sing}}_1(k_1)$ is nonzero only for $k_1=1$, in which case it is equal to $2(2p+1)u^{2p+1} \frac{1}{p-\frac12}$. 
\end{proposition}

\begin{proof} First, note that the shift of the degree in $u$ in the definition of $A^{(g)}_n$ \eqref{eq:A-definiton} is compensated by the shift of degree in $u$ in the change of variables~\eqref{eq:t-tau-relation}. Thus, the non-negative even degrees of $u$ in $A^{(g)}_n$ are coming from the non-negative even degrees of $u$ in $F^{(g)}$ and vice versa. 
Then notice that 
\begin{align}
\frac{\partial^n F^{(g)}}{\partial t_{k_1}\cdots  \partial t_{k_n}}\Big|_{t_l = t_l^o,\ l=1,\dots,p}
\end{align}
is a monomial in $u$ of degree $D = (2-2g)p -(3g-3+n) - \sumk_{\set{n}}$, as in~\eqref{eq:degree}.

For $n=1$, we expand~\eqref{eq:01-case} as
\begin{align} \label{eq:01-case-expanded}
	2\frac{u^{2p+2-k}}{p-k+\frac 12} \Res_{w=0} 
	\frac{(1+w^4)^{p-k+\frac 12}(1+w^{4p+2})(1-w^4)}{w^{4p-2k+4}} \frac {dw}w,
\end{align}
and we see that the factor of $w^{4p+2}$ contributes only for $k=1$, which is exactly the special case in the statement of the lemma. Otherwise, if $k\not=1$, $2p-2k+n$ must be a multiple of $4$, hence $k$ must be even, and thus the degree of $u$ equal to $2p+2-k$ is a positive even number. 

For $n=2$, we notice that $B(z_1,z_2) = B(w_1,w_2)- dw_1dw_2/(1-w_1w_2)^2$ and expand~\eqref{eq:02-case} as
\begin{align}\label{eq:02-case-expanded}
	& 
	\frac{u^{2p+1-k_1-k_2}}{(p-k_1+\frac 12)(p-k_2+\frac12)}\Res_{w_1=0} \Res_{w_2=0} \frac{(1+w_1^4)^{{p-k_1+\frac 12}}(1+w_2^4)^{{p-k_2+\frac 12}}}{w_1^{2p-2k_1+1}w_2^{2p-2k_2+1}} \sum_{d=0}^\infty (d+1) w_1^dw_2^{-d-2} dw_1dw_2
	\\ \notag & 
	-
	\frac{u^{2p+1-k_1-k_2}}{(p-k_1+\frac 12)(p-k_2+\frac 12)}\Res_{w_1=0} \Res_{w_2=0} \frac{(1+w_1^4)^{{p-k_1+\frac 12}}(1+w_2^4)^{{p-k_2+\frac 12}}}{w_1^{2p-2k_1+1}w_2^{2p-2k_2+1}} \sum_{d=0}^\infty (d+1) w_1^dw_2^{d} dw_1dw_2.
\end{align}
Note that in the first line the non-trivial residue is only possible if $4p-2k_1-2k_2+2$ is a multiple of $4$. Thus $k_1+k_2$ is odd and $2p+1-k_1-k_2$ is even, hence the whole first summand is non-singular in $u^2$.
\end{proof}

\begin{remark} Note that the suggested adjustment of the definition of $F^{(g)}$ that defines $F^{(g),\mathsf{sing}}$ is very natural in the context of KP integrability in the theory of topological recursion. Indeed, when one passes to a KP tau function by expansion in a non-singular point in some local coordinate $w$, 
the replacements $\omega^{(0)}_1\rightsquigarrow 0$ and $\omega^{(0)}_2 \rightsquigarrow \omega^{(0)}_2-B(w_1,w_2)$ are the standard choices, cf.~\cite[Rem.~2.4 and Eqs.~(2.5),~(2.6)]{ABDKS-KPxy}.
\end{remark}

\subsection{Final form of Artemev's conjecture} Now we can formulate a theorem that settles Artemev's conjecture including also possible effects of (stabilized) non-singular terms and suitable renormalization of the involved constants. 

\begin{theorem}\label{thm:main} For any $g\geq 0$, $n\geq 1$, $k_1,\dots,k_n=1,\dots,p$, we have
\begin{align}
	& \frac {1}{2^n(2p+1)^n} A^{(g),\mathsf{st}}_n(k_1,\dots,k_n) + \delta_{g,0}\delta_{n,1} \frac{u^{2p+1}}{p-\frac 12}  \\ \notag & = A^{(g),\vee}_n(k_1,\dots,k_n)  -	\delta_{g,0} \delta_{n,\geq 2} \prod_{i=1}^n \Bigg( \Res_{w_i=0} \frac{2(2p+1)u^{p-k_i+\frac 12}}{2(p-k_i+\frac 12) w^{2p-2k_i+1}}  \Bigg)
	\omega^{(0),\sim}_n,
\end{align}
where the second summand is a polynomial in $u^2$ for each $n\geq 2$. It is expressed in terms of the differentials $\omega^{(0),\sim}_n$ of generalized topological recursion in the sense of~\cite{ABDKS-gen} for the spectral curve data given by $\mathbb{C}P^1$, with the standard Bergman kernel $B=dw_1dw_2/(w_1-w_2)^2$ in the global affine coordinate $w$, the initial differentials
	\begin{align} \label{eq:GenTRspectral-1sttime}
	dx^\sim & = u^{p+\frac 12} d\Big( \frac 1 {w^{2p+1}}\Big),
	& dy^\sim & =  ud\Big(w^2 + \frac 1 {w^2}\Big),
\end{align}
and the maximal possible set of key points.
\end{theorem}

\begin{proof} The proof consists of the computations performed in the next two sections, where we give an explicit computation of the left hand side and the right hand side of this equality in terms of $\omega^{(g)}_\ell$, $1\leq \ell \leq n$. In particular, for $(g,n)\not=(0,1)$ one can directly match the left hand side computed in Prop.~\ref{prop:resonancetr} and the right hand side computed in Prop.~\ref{prop:xydual}, respectively. The case $(g,n)=(0,1)$ is special, and we have to reinstall it by hand on the left hand side as required by Prop.~\ref{prop:sing-stable}, but on the right hand side it is automatically included by the virtue of the $x-y$ swap formula, see Rem.~\ref{rem:01-xydual}.
\end{proof}

\begin{remark} Note that the statement of Thm.~\ref{thm:main} is, on the one hand, more precise than Conj.~\ref{conj:art}, on the other hand it is also slighly weaker than Conj.~\ref{conj:art}. Namely, while we indeed have a very good control and full description of how the polynomial terms in $u^2$ change when we move from the resonance transformation for the minimal string spectral curve to the $x-y$ swap side, it is also suggested in~\cite{artemev} that $A^{(g),\vee}_n$ is purely singular in $u^2$. It is something that our approach doesn't allow to derive. 
\end{remark}

\section{Explicit computation via resonance transformations}

The goal of this secton is to derive an explicit formula for 	$A^{(g),\mathsf{st}}_n(k_1,\dots,k_n)$ in terms of symmetric differentials $\overline{\omega}^{(g)}_n$ defined as 
\begin{align} \label{eq:overline-omega-def}
	\overline{\omega}^{(0)}_1& \coloneqq 0; 
	&\overline{\omega}^{(0)}_2 & \coloneqq {\omega}^{(0)}_2 - B(w_1,w_2)
	 = \frac{dw_1dw_2}{(1-w_1w_2)^2}; 
	& \overline{\omega}^{(g)}_n & \coloneqq {\omega}^{(g)}_n,\   2g-2+n>0. 
\end{align}
Note that all these differentials are holomorphic at $w=0$ in each variable, so the formula
\begin{align} \label{eq:DefFg-st}
		F^{(g),\mathsf{st}}(t_1,\dots,t_p) 
		& = \sum_{n=1}^\infty \frac{(-1)^n}{n!} \sum_{k_1,\dots,k_n=1}^p \prod_{i=1}^n \Bigg( (t_{k_i}-t_{k_i}^o) \Res_{w_i=0}\frac{x(w_i)^{p-k_i+\frac 12}}{(p-k_i+\frac 12)} \Bigg) \overline{\omega}^{(g)}_n
\end{align}
doesn't need any discussion of any exceptonal cases. The main statement in the following:
\begin{proposition} \label{prop:resonancetr} For any $g\geq 0$, $n\geq 1$, $k_1,\dots,k_n=1,\dots,p$ we have:
		\begin{align} \label{eq:prop:resonancetr}
			A^{(g),\mathsf{st}}_n(k_1,\dots,k_n) & = u^{-\sumk_{\set{n}}-n} \sum_{\ell=1}^n \frac{(-2)^\ell (2p+1)^\ell u^{(p+\frac 32)\ell}}{\ell!} \sum_{\substack{I_1\sqcup \cdots \sqcup I_\ell = \set{n} \\ I_1,\dots,I_\ell \not=\emptyset}} \prod_{i=1}^\ell \Bigg(\Res_{w_i=0} \sum_{s_i=0}^{\infty} \frac{(|I_i|-1)_{s_i}}{s_i!} \times
			\\ \notag & \qquad 
			(-2)^{|I_i|-2}\Big(-p+\sumk_{I_i}+s_i+\frac 12\Big)_{|I_i|-2}
			\frac{1}{w_i^{2p-2\sumk_{I_i} -2|I_i|-4s_i+3}} \Bigg) \overline{\omega}^{(g)}_\ell
		\end{align}
\end{proposition}

\begin{proof} First, by explicit substitution into~\eqref{eq:t-tau-relation} we derive that for $\sumk_{\set{m}}+m-1=k$
\begin{align}
		{\frac{\partial^m t_l}{\partial\tau_{k_1}\cdots\partial\tau_{k_m}}}\Big|_{\tau_l = \tau_l^o,\ l=1,\dots,p}=
			(2p+1)u^{l+1}(-2)^{m-1}\frac{(\frac{l-k}{2}+m-1)!}{(\frac{l-k}{2})!}
			\binom{-p+l-\frac{1}{2}}{\frac{l-k}{2}+m-1} 
\end{align}
if $l - k$ is even, and it is equal to zero otherwise. Then, notice that at $w\to 0$
\begin{align}
	x^{p-l+\frac{1}{2}}(w)&=u^{p-l+\frac{1}{2}}\sum_{j=0}^{\infty}\binom{p-l+\frac{1}{2}}{j}\frac{1}{w^{2p-2l-4j+1}}.
\end{align}
Combining these two equations with the chain rule, we obtain the following formula for 		$A^{(g),\mathsf{st}}_n(k_1,\dots,k_n)$:
\begin{align} \label{eq:Ag-st-intermediate}
	A^{(g),\mathsf{st}}_n(k_1,\dots,k_n) & = u^{-\sumk_{\set{n}}-n}\sum_{\ell=1}^n \frac{1}{\ell!} \sum_{\substack{I_1\sqcup \cdots \sqcup I_\ell = \set{n} \\ I_1,\dots,I_\ell \not=\emptyset}} \prod_{i=1}^\ell \Bigg(\Res_{w_i=0} f(\sumk_{I_i}+|I_i|-1,|I_i|,w_i) \Bigg) \overline{\omega}^{(g)}_\ell,
\end{align}
where
\begin{align} \label{eq:formula-for-f}
	f(k,m,w) & \coloneqq 
 -2(2p+1)u^{p+3/2}\sum_{\substack {l=k \\ 2\,|\,l-k}}^p\sum_{j=0}^{\infty}(-2)^{m-1}\frac{(\frac{l-k}{2}+m-1)!}{(\frac{l-k}{2})!}\binom{-p+l-\frac{1}{2}}{\frac{l-k}{2}+m-1}
 \times \\ \notag & \qquad \qquad \qquad \qquad \qquad \qquad  \frac{1}{2(p-l+\frac{1}{2})}\binom{p-l+\frac{1}{2}}{j}\frac{1}{w^{2p-2l-4j+1}}.
\end{align}
Note that the contributions to $\Res_{w=0}f(k,m,w) \omega(w)$ for a differential $\omega$ holomorphic at $w\to 0$ are nontrivial only for $j\leq \lfloor \frac{p-l}{2} \rfloor$ (this condition selects the principal part, that is, the negative degrees in $w$, of the expansion of $f(k,m,w)$ at $w\to 0$). 

Lemma~\ref{lem:1stCombLemma} below simplifies the expression for the principal part of $f(k,m,w)$, see~\eqref{eq:simplified-f} below. Substituting~\eqref{eq:simplified-f} into~\eqref{eq:Ag-st-intermediate} we obtain the statement of the proposition. 
\end{proof}

\begin{lemma} \label{lem:1stCombLemma} The principal part of $f(k,m,w)$ at $w\to 0$ is equal to the principal part of the following expression:
	\begin{align} \label{eq:simplified-f}
		-2(2p+1)u^{p+3/2}
		\sum_{s=0}^{\infty 
		}\frac{(-1)^m}{s!}2^{m-2}(m-1)_s\left(-p+k+s-m+\frac{3}{2}\right)_{m-2}\frac{1}{w^{2p-2k-4s+1}}.
	\end{align}
\end{lemma}

\begin{proof}  We leave out the common factor of $-2(2p+1)u^{p+3/2}$. Consider~\eqref{eq:formula-for-f} with extra condition $j\leq \lfloor \frac{p-l}{2} \rfloor$. Set $l=k+2a$ with $a\in\{0,\dots,\lfloor\frac{p-k}{2}\rfloor\}$. Then the powers of $w$ can be written as $-({2p-2k-4a-4j+1})=-({2p-2k-4s+1})$ with $s=a+j$. Taking $s$ fixed we see $j=s-a$ and then $a$ lies in $\{0,\dots,\min\{s,\lfloor\frac{p-k}{2}\rfloor\}\}$. In our case we always have $\min\{s,\lfloor\frac{p-k}{2}\rfloor\}=s$ as choosing $a>\lfloor\frac{p-k}{2}\rfloor$ gives a positive power of $w$. Set $c=p-k+\frac{1}{2}$. The coefficient of $w^{-(2p-2k-4s+1)}$ is then
	\begin{align}
		&\sum_{a=0}^s(-2)^{m-1}\frac{(a+m-1)!}{a!}\frac{1}{2c-4a}\binom{2a-c}{a+m-1}\binom{c-2a}{s-a}\\ \notag 
		&=\sum_{a=0}^s(-2)^{m-1}\frac{(a+m-1)!}{a!}\frac{1}{2c-4a}(-1)^{a+m-1}\binom{c-a+m-1}{a+m-1}\binom{c-2a}{s-a}\\ \notag 
		&=\sum_{a=0}^s2^{m-1}(-1)^a\frac{1}{2c-4a}\frac{\Gamma(a+m)}{\Gamma(a+1)}\frac{\Gamma(c-a+m-1)}{\Gamma(a+m)\Gamma(c-2a)}\frac{\Gamma(c-2a+1)}{\Gamma(s-a+1)\Gamma(c-a-s+1)}.
	\end{align}
	Now we can use the property $\Gamma(x+1)=x\Gamma(x)$ to write $\frac{1}{2c-4a}\frac{\Gamma(c-2a+1)}{\Gamma(c-2a)}=\frac{c-2a}{2c-4a}=\frac{1}{2}$. The sum thus reduces to
	\begin{align}
		\sum_{a=0}^s \frac{2^{m-2}(-1)^a \Gamma(c-a+m-1)}{\Gamma(a+1)\Gamma(s-a+1)\Gamma(c-a-s+1)}
		&=\frac{1}{s!}\sum_{a=0}^s2^{m-2}(-1)^a\binom{s}{a}\frac{\Gamma(c-a+m-1)}{\Gamma(c-a-s+1)}\\ \notag 
		&=\frac{1}{s!}\sum_{a=0}^s2^{m-2}(-1)^a\binom{s}{a}(c-a-s+1)_{s+m-2}.
	\end{align}
	Now use the finite difference identity $\sum_{a=0}^s(-1)^a\binom{s}{a}(t-a)_N=(N-s+1)_s(t)_{N-s}$ to rewrite this as
	\begin{equation}
		\frac{1}{s!}2^{m-2}(m-1)_s(c-s+1)_{m-2}.
	\end{equation}
	Using the reflection identity $(-1)^N(t)_N=(-t-N+1)_N$ and plugging $c$ in finally gives
	\begin{equation}
		\frac{(-1)^m}{s!}2^{m-2}(m-1)_s\Big(-p+k+s-m+\frac{3}{2}\Big)_{m-2}.
	\end{equation}
\end{proof}

\section{Explicit computation on the dual side}

\subsection{Analysis of the \texorpdfstring{$x-y$}{x-y} swap formula}

Recall the universal formula for the $x-y$ swap in the theory of topological recursion~\cite[Eq.~(1.13)]{alexandrov2022universal} adapted to our conventions. We have:
\begin{align} \label{eq:MainFormulaSimple}
	 \omega_{n}^{(g),\vee} (z_{\set{n}})
	& =
	(-1)^n
	\coeff {\hbar^{2g}} \sum_{\Gamma} \frac{\hbar^{2g(\Gamma)}}{|\mathrm{Aut}(\Gamma)|} \prod_{i=1}^n dy_i
	\sum_{k_i=0}^\infty  (-\partial_{y_i})^{k_i} \coeff {v_i^{k_i}}  \frac{1}{dy_i} \frac{dx_i}{v_i}
	\\ \notag & \qquad \qquad 
	 \exp \bigg( {  \mathcal{S}(\hbar v_i \partial_{x_i}) \sum_{\tilde g=0}^\infty \hbar^{2\tilde g} \frac{v_i}{dx_i}\omega^{(\tilde g)}_{1} (z_i)- \frac{v_i}{dx_i}\omega^{(0)}_{1} (z_i)} \bigg)
	\\ \notag & \qquad \qquad 
	\prod_{e\in E(\Gamma)} \prod_{j=1}^{|e|\geq 2}\restr{(\tilde v_j, \tilde x_j)}{ (v_{e(j)},x_{e(j)})} \mathcal{S}(\hbar \tilde v_j  \partial_{\tilde x_j}) \sum_{\tilde g=0}^\infty \hbar^{2\tilde g}\tilde \omega^{(\tilde g)}_{|e|}(\tilde z_{\llbracket |e|\rrbracket})  \prod_{j=1}^{|e|} \frac{ \tilde v_j}{d\tilde x_j}
	\\ \notag & \quad 
	+\delta_{(g,n),(0,1)} x_1dy_1.
\end{align}
Here the following conventions are used:
\begin{itemize}
	\item We use notation $x_i \coloneqq x(z_i)$, $\tilde x_i \coloneqq x(\tilde z_i)$, $y_i\coloneqq y(z_i)$.
	\item The sum is taken over all connected graphs $\Gamma$ with $n$ labeled vertices (labelled by $i\in \set{n}$) and multiedges of index $\geq 2$. 
	By $g(\Gamma)$ we denote the first Betti number of $\Gamma$, and $|\mathrm{Aut}(\Gamma)|$ stands for the number of automorphisms of $\Gamma$.
	\item For convenience, we also label all legs of every given multiedge $e$ from $1$ to $|e|$ in an arbitrary way. This is not a part of the data of the graph and these labels don't affect $|\mathrm{Aut}(\Gamma)|$.
	\item For a multiedge $e$ with index $|e|$  we control its attachment to the vertices by the associated map $e\colon \llbracket |e| \rrbracket \to 
	\llbracket n \rrbracket
	$ that we denote also by $e$, abusing notation. 
	\item For a multiedge $e$ with $|e|=2$ we define $\tilde \omega^{(0)}_{2} :=  \omega^{(0)}_{2} - \frac{d\tilde x_1d\tilde x_2}{(\tilde x_1-\tilde x_2)^2}$ if $e(1)=e(2)$, and $\tilde \omega^{(0)}_{2} :=  \omega^{(0)}_{2}$ otherwise. For all $(g,n)\not=(0,2)$ we simply have $\tilde \omega^{(g)}_{n} :=  \omega^{(g)}_{n}$.
	\item By $\coeff{\hbar^{2g}}$ (respectively, $\coeff{v_i^{p_i}}$) we denote the operator that extracts the corresponding coefficient from the whole expression to the right of it.
	\item By $\restr{a}{b}$ we denote the operator of substitution, $\restr{a}{b}f(a)\coloneqq f(b)$ for any function $f$.
	\item The function $\mathcal{S}(\zeta)$ is defined as $\mathcal{S}(\zeta)\coloneqq \frac {e^{\zeta/2}-e^{-\zeta/2}} \zeta$.
\end{itemize}
It will be convenient for the argument below to apply a  slight reshuffeling of this formula, namely, for a multiedge $e$ of index $2$ with $e(1)\not= e(2)$ we split the associated 2-differential as 
\begin{align} 
\mathcal{S}(\hbar \tilde v_1 \partial_{\tilde x_1}) 
 \mathcal{S}(\hbar \tilde v_2 \partial_{\tilde x_2})\frac{\tilde v_1\tilde v_2\big(\sum_{\tilde g=0}^\infty \hbar^{2\tilde g}\tilde \omega^{(\tilde g)}_{2}(\tilde z_1, \tilde z_2) -  B(\tilde w_1,\tilde w_2) \big)}{\tilde dx_1\tilde dx_2}  + \mathcal{S}(\hbar \tilde v_1 \partial_{\tilde x_1}) \mathcal{S}(\hbar \tilde v_2 \partial_{\tilde x_2}) \frac{\tilde v_1\tilde v_2 B(\tilde w_1,\tilde w_2)}{\tilde dx_1\tilde dx_2},
\end{align}
and thus we get two types of 2-multiedges: the ``stable'' ones labeled by the first summand and the ``unstable'' ones labeled by $\mathcal{S}(\hbar \tilde v_1 \partial_{\tilde x_1}) \mathcal{S}(\hbar \tilde v_2 \partial_{\tilde x_2}) \frac{\tilde v_1\tilde v_2 B(\tilde w_1,\tilde w_2)}{\tilde dx_1\tilde dx_2}$. 
 It is also convenient to regard the terms in the exponent in the second line of~\eqref{eq:MainFormulaSimple} as multiedges of index $1$, where we also split the labels into the ``stable'' ones  $\mathcal{S}(\hbar v_i \partial_{x_i}) \sum_{\tilde g=1}^\infty \hbar^{2\tilde g} \frac{v_i}{dx_i}\omega^{(\tilde g)}_{1} (z_i)$ and ``unstable'' ones $(\mathcal{S}(\hbar v_i \partial_{x_i})-1)  \frac{v_i}{dx_i}\omega^{(\tilde 0)}_{1} (z_i)$. The multiedges of index $\geq 3$ we by default consider to be ``stable''.

For our purpose of computation of  $A^{(g),\vee}_n$, define $\omega^{(g),\vee\!\!\!\vee}_n$ as a severly reduced version of $\omega^{(g),\vee}_n$:
\begin{align} \label{eq:omega-vee-one-stable}
	\omega_{n}^{(g),\vee\!\!\!\vee} (w_{\set{n}})
	& \coloneqq 
	(-1)^n \sum_{\substack {\Gamma:\,g(\Gamma)=0 \\ E(\Gamma) = \{e_o\}\sqcup E'(\Gamma)}} \prod_{i=1}^n \widehat{dy}_i
	\sum_{k_i=0}^\infty  (-\widehat{\partial_{y_i}})^{k_i} \coeff {v_i^{k_i}}  \frac{1}{\widehat{dy}_i} \frac{{dx}_i}{v_i} \times
	\\ \notag &  \qquad 
	\overline{\omega}^{(g)}_{|e_o|}(w_{e_o(1)},\dots,w_{e_o(|e_o|)}) \prod_{j=1}^{|e_o|} \frac{v_{e_o(j)}}{{dx}_{e_o(j)}}
	\prod_{e\in E'(\Gamma)} B(w_{e(1)},w_{e(2)})  \frac{ v_{e(1)}}{{dx}_{e(1)}}\frac{ v_{e(2)}}{{dx}_{e(2)}}.
\end{align}
Here the following new additional conventions are used:
\begin{itemize}
	\item The sum is taken over all connected graphs $\Gamma$ of genus $0$ with $n$ labeled vertices. These graphs have no automorphisms.
	There is one distinguished multiedge $e_o$ that might be of any index $|e_o|\geq 1$ and all other multiedges are of index $2$. Note that every multiedge of index $\ell$ is attached to $\ell$ different vertices. 
	\item The differentials assigned by multiedges are defined in~\eqref{eq:overline-omega-def}. 
	\item The differential $\widehat{dy}$ and the corresponding derivative are defined as
	\begin{align}\label{eq:hat-operators}
		\widehat{dy} & = - \frac{(2p+1)u^{p+\frac 12}dw}{w^{2p+2}} ; & \widehat{\partial_y}  & = - \frac{w^{2p+2}}{(2p+1)u^{p+\frac 12}}\partial_w.
	\end{align}
\end{itemize}

Recall the definition of $A^{(g),\vee}_n(k_1,\dots,k_n)$ in~\eqref{eq:defA-vee}. We have the following lemma:
\begin{lemma} \label{lem:vee-to-veevee} For any $g\geq 0$, $n\geq 1$, $(g,n)\not=(0,1)$, $k_1,\dots,k_n = 1,\dots,p$, we have:
\begin{align} \label{eq:simplification-of-graphs}
 \prod_{i=1}^n \Bigg( \Res_{w_i=0} \frac{\big(x^\vee(w_i)\big)^{\frac{2p-2k_i+1}{2p+1}}}{2p-2k_i+1} \Bigg) \omega^{(g),\vee}_n 
= \prod_{i=1}^n \Bigg( \Res_{w_i=0} \frac{u^{p-k_i+\frac 12}}{2(p-k_i+\frac 12) w^{2p-2k_i+1}}  \Bigg)
\omega^{(g),\vee\!\!\!\vee}_n
\end{align} 
up to an extra term of even positive degree in $u$ that appears only for $g=0$ and whose explicit description is given below, in Rem.~\eqref{rem:sim-explanation}. 
\end{lemma}

\begin{proof} Note that ~\eqref{eq:MainFormulaSimple} defines for $2g-2+n>0$ an $n$-differential that has no poles on the diagonals $z_i=z_j$, despite the fact that individual summand do have them. This means that the right way to understand~\eqref{eq:simplification-of-graphs} is to fix the consecutive order of the residues $\prod_{i=1}^n \Res_{w_i=0} = \Res_{w_n=0} \cdots \Res_{w_1=0}$ and once we apply them in this order to the individual decorated graphs, fix the expansions of the singular terms in the sector $|w_1|<\cdots < |w_n|$. 
	
We select the contributing graphs in~\eqref{eq:MainFormulaSimple} using rough estimation of the orders of zeros at $w_i=0$, $i=1,\dots,n$. Note that 
\begin{align}
	{dx} & = - \frac{2u (1-w^4)dw}{w^3} ; & \partial_x  & = - \frac{w^3}{2u (1-w^4)} \partial_w; \\
	{dy} & = - \frac{(2p+1)u^{p+\frac 12} (1-w^{4p+2})dw}{w^{2p+2}} ; & \partial_y  & = - \frac{w^{2p+2}}{(2p+1)u^{p+\frac 12} (1-w^{4p+2})}\partial_w.
\end{align}
This means that each occurence of $v_i/dx_i$ increases the order of zero in $w_i$ at $w_i=0$ by at least $(2p+4)$. Also, each occurence of $v_i\partial_{x_i}$ that enter the $\mathcal{S}$-series increases the order of zero in $w_i$ at $w_i=0$ by at least $(2p+3)$. Note that the order of zero in $w_i$ can be negative if we have a pole --- this might be the case if we have multiedges of index $2$ attached to the $i$-th vertex as they contain the suitably expanded Bergman kernel.
Moreover, 
\begin{align}
\omega^{(0)}_1(w) =  -2	u^{p+\frac 32} \frac {(1+w^{4p+2})(1-w^4)}{w^{2p+4}} dw,
\end{align}
but it enters the formula with $(\mathcal{S}(\hbar v\partial_x)-1)$ applied to it, which shifts the order of zero from $-(2p+4)$ to  at lease $2p+2$.
 On the other hand, for $k=1,\dots,p$ the operator
\begin{align} \label{eq:res-operator}
\Res_{w=0} \frac{\big(x^\vee(w)\big)^{\frac{2p-2k+1}{2p+1}}}{2p-2k+1} = \Res_{w=0} \frac{u^{p-k+\frac 12}(1+w^{4p+2})^{\frac{2p-2k+1}{2p+1}}}{2(p-k+\frac 12) w^{2p-2k+1}}
\end{align}
vanishes on differentials with the order of zero at $w=0$ greater or equal to $2p$. Final observation that we need is that the term 
\begin{align}
	B(w_i,w_j)= dw_1dw_2 \sum_{d=0}^\infty \frac{w_i^d}{w_j^{d+2}} (d+1)
\end{align}
increases the order of zero by $d$ at $w_i=0$ simultaneously decreasing the order of zero by $d+2$ at $w_j=0$. 

Since the expression for vertex $i$ of index $p$ in~\eqref{eq:MainFormulaSimple} has at least $p-1$ factors $\frac{v_i}{dx_i}$, it should be either of index $1$ or at least one of the attached multiedges should be a ``polar'' contribution from an unstable $2$-multiedge. In fact, the ``polar'' contribution in combination with the extra factor of $\frac{v_i}{dx_i}$ that comes with it still increases the degree of zero at $w_i=0$, but this shift can be less than $2p$, hence the residue operator can still be nontrivial. Note also that $\mathcal{S}(\hbar v \partial_x) = 1+O(\hbar v \partial_x)^2$, that is, its non-trivial part increases the order of zero by at least $2p+6$. Hence, all instances of $\mathcal{S}(\hbar v \partial_x)$ must be replaced just by $1$ (its leading term), which in particular kills the unstable multiedges of index $1$. The same holds for the factor $(1-w^{4p+2})$ in the formulas for $dy$ and $\partial_y$, hence we replace them by $\widehat{dy}$ and $\widehat{\partial_y}$.

This implies that we have non-trivial contribution only from the graphs of total genus $0$, with at most one stable multiedge (hence decorated by $\overline{\omega}^{(g)}_{|e_o|}(z_{e_o(1)},\dots,z_{e_o(|e_o|)}) \prod_{j=1}^{|e_o|} \frac{v_{e_o(j)}}{d x_{e_o(j)}}$), and with any number of unstable multiedges of index $2$ decorated by $B(w_{e(1)},w_{e(2)})  \frac{ v_{e(1)}}{d x_{e(1)}}\frac{ v_{e(2)}}{d x_{e(2)}}$. Note also that the terms of degree $4p+2$ and higher in the numerator of~\eqref{eq:res-operator} increase the order of zero beyond the order of the pole and thus don't contribute. Thus we obtain the sum of the right hand side of~\eqref{eq:simplification-of-graphs} (in the case there is exactly one stable multiedge) and the following term that represents the case of no stable multiedges at all:
\begin{align} \label{eq:extraterm}
	 \prod_{i=1}^n \Bigg( \Res_{w_i=0} \frac{u^{p-k_i+\frac 12}}{2(p-k_i+\frac 12) w^{2p-2k_i+1}}  \Bigg)
	\big(\omega^{(0),\sim}_n	+\delta_{n,1} x_1dy_1\big),
\end{align} 
where
\begin{align} \label{eq:omega-sim}
	\omega_{n}^{(0),\sim} (w_{\set{n}})
	& \coloneqq 
	(-1)^n \sum_{\substack {\Gamma:\,g(\Gamma)=0}} \prod_{i=1}^n \widehat{dy}_i
	\sum_{k_i=0}^\infty  (-\widehat{\partial_{y_i}})^{k_i} \coeff {v_i^{k_i}}   \frac{dx_i}{v_i \widehat{dy}_i}
\prod_{e\in E(\Gamma)} \frac{v_{e(1)}v_{e(2)}B(w_{e(1)},w_{e(2)})}{d x_{e(1)}d x_{e(2)}} 
\end{align}
 The second summand corresponds to the exceptional $(0,1)$ term and it is not a part of the statement of the lemma (see Rem.~\ref{rem:01-xydual} below). As for the first summand, the rough count of the orders of zero at $w_i=0$, $i=1,\dots,n$, gives the following. There are $n-1$ Bergman kernels, and each of them adds $-2$ to the total sum of orders of zeros. There are also $n-2$ extra factors of $v/dx$, and each of them adds $2p+4+4s$ to the total sum of orders of zeros, for some $s\in \mathbb{Z}_{\geq 0}$. On the other hand, we apply the product of the operators~\eqref{eq:res-operator}, and this implies that the total sum of orders of zeros must be $2pn-2\sumk_{\set{n}}$. Thus we get the following equation for possibly non-trivial contributions:
 \begin{align}
 	2pn-2\sumk_{\set{n}} = -2(n-1)+(n-2)(2p+4)+4s,
 \end{align}
 which implies that $\sumk_{\set{n}}=2p-n+3=2s$ and the corresponding degree of $u$ is equal to $2p+3-n-\sumk_{\set{n}} = 2s \in 2\mathbb{Z}_{\geq 0}$. 
\end{proof}

\begin{remark} \label{rem:01-xydual} In the case $(g,n)=(0,1)$ a straightforward  computation gives 
\begin{align}
		 \Res_{w=0} \frac{\big(x^\vee(w)\big)^{\frac{2p-2k+1}{2p+1}}}{2p-2k+1} \omega^{(0),\vee}_1 = \begin{cases}
		 	u^{2p+1} \frac{1}{p-\frac 12}, & k=1;
		 	\\
		 	0, & k=2,\dots,p.
		 \end{cases}
\end{align} 
\end{remark}

\begin{remark} \label{rem:sim-explanation} Note that $\omega_{n}^{(0),\sim} (w_{\set{n}})$ has an independent intrinsic meaning in terms of the \emph{generalized topological recursion} in the sense of~\cite[Def.~2.8]{ABDKS-gen}. Namely, for the spectral curve
$\mathbb{C}\mathrm{P}^1$ with the global affine coordinate $w$, the standard Bergman kernel $B(w_1,w_2)$ and the differentials $dx^\sim$ and $dy^\sim$ given by
	\begin{align} \label{eq:GenTRspectral}
		dx^\sim & = u^{p+\frac 12} d\Big( \frac 1 {w^{2p+1}}\Big),
		& dy^\sim & =  ud\Big(w^2 + \frac 1 {w^2}\Big),
	\end{align}
the set of special points is given by all $q\in \mathbb{C}\mathrm{P}^1$ s.t. either $dy^\sim|_q=0$ or $w(q)=\infty$. Choose as the input the set of key point $\mathcal{P}$ to coincide with the set of the special points. Then~\eqref{eq:omega-sim} gives a close formula for the genus $0$ differentials of generalized topological recursion in this case. To see this notice that the generalized topological recursion for the $x-y$ dual initial data is trivial, and in genus $0$ the $x-y$ swap formula gives exactly~\eqref{eq:omega-sim}, cf.~\cite[Sec.~7.1]{ABDKS-gen}.
\end{remark}

\subsection{Recursive structure of residues} Now our goal is to compute the right hand side of~\eqref{eq:simplification-of-graphs}. To this end we notice that once we remove the distinguished multiedge $e_o$ from the graphs in the definition of $\omega^{(g),\vee\!\!\!\vee}_n$, they split into $|e_o|$ components, and the computation of the residues in~\eqref{eq:simplification-of-graphs} is independent for each of these components. Thus, it is sufficient to consider the following situation:

Let $\overline{\omega}$ be a $1$-differential expanded in a Laurent series at $w=0$ (we won't assume that it is regular at $w=0$ since our inductive argument below will involve more general setup). Let $\omega_{m}^{\vee\!\!\!\vee} (w_{\set{m}})$ be an $m$-differential defined as
\begin{align} \label{eq:omega-veevee-abstract}
	\omega_{m}^{\vee\!\!\!\vee} (w_{\set{m}})
	& \coloneqq 
	(-1)^m \sum_{\substack {\Gamma:\,g(\Gamma)=0}} \prod_{i=1}^m \widehat{dy}_i
	\sum_{k_i=0}^\infty  (-\widehat{\partial_{y_i}})^{k_i} \coeff {v_i^{k_i}}  \frac{1}{\widehat{dy}_i} \frac{dx_i}{v_i}
	\\ \notag &  \qquad \qquad 
	\prod_{e\in E(\Gamma)} B(w_{e(1)},w_{e(2)})  \frac{ v_{e(1)}}{d x_{e(1)}}\frac{ v_{e(2)}}{d x_{e(2)}}
	\sum_{i=1}^m	\overline{\omega}(w_{i}) \frac{v_{i}}{d x_{i}}.
\end{align} 
Here the sum is taken over all trees with $m$ labeled vertices with usual edges. 

\begin{lemma} \label{lem:veevee-tooperator} For any $k_1,\dots,k_m = 1,\dots,p$, we have:
	\begin{align} \label{eq:residue-tree}
		& u^{\sumk_{\set{m}}+m} \prod_{i=1}^m \Bigg( \Res_{w_i=0} \frac{u^{p-k_i+\frac 12}}{2(p-k_i+\frac 12) w^{2p-2k_i+1}}  \Bigg)
		\omega^{\vee\!\!\!\vee}_m = 	- (2(2p+1))^{1-m}
		u^{p+3/2} \times
		\\ \notag & \qquad 
	\Res_{w=0}
		\sum_{s=0}^{\infty}\frac{(-2)^{m-2}}{s!}(m-1)_s\left(-p+\sumk_{\set{m}}+s+\frac{1}{2}\right)_{m-2}\frac{\overline{\omega}(w)}{w^{2p-2\sumk_{\set{m}}-2m-4s+3}}.
	\end{align} 
\end{lemma}

\begin{proof} For $m=1,2$ it can be checked by an explicit computation. For general $m$ we prove the lemma by induction on the number of vertices of the tree. Assume that the vertex where $\overline{\omega}$ is places has index $\ell$, and without loss of generality assume that it is labeled by $m$ (but we don't assume that it is the last index in the fixed order of residues). Then, when we remove this vertex from the tree, the tree splits into $\ell$ smaller trees whose vertices are labeled by $I_1,\dots,I_\ell$ such that $I_1\sqcup \cdots \sqcup I_\ell = \set{m-1}$. By induction, the part of the residue formula for each of these smaller trees gives 
\begin{align}
	& - (2(2p+1))^{1-|I_i|}u^{p+3/2} \Res_{w'=0}
	\sum_{s_i=0}^{\infty}\frac{(-2)^{|I_i|-2}}{(s_i)!}(|I_i|-1)_{s_i}\left(-p+\sumk_{I_i}+s_i+\frac{1}{2}\right)_{|I_i|-2}\times
	\\ \notag & 
	\qquad \qquad \frac{1}{(w')^{2p-2\sumk_{I_i}-2|I_i|-4s_i+3}}B(w',w)
	\\ \notag & 
	= - (2(2p+1))^{1-|I_i|} u^{p+3/2}
	\sum_{s_i=0}^{\infty}\frac{(-2)^{|I_i|-2}}{(s_i)!}(|I_i|-1)_{s_i}\left(-p+\sumk_{I_i}+s_i+\frac{1}{2}\right)_{|I_i|-2} \times
	\\ \notag & \qquad \qquad (2p-2\sumk_{I_i}-2|I_i|-4s_i+3) \frac{dw}{w^{2p-2\sumk_{I_i}-2|I_i|-4s_i+4}},
\end{align}
where the degree of $w$ might be either positive or negative---this depends on $s_i$ and it determines the sector of expansion as well as the relative position to which $\overline{\omega}$ is attached. Thus the left hand side of~\eqref{eq:residue-tree} can be computed as 
\begin{align} \label{eq:combxcompdualside}
	& \Res_{w=0} \frac{u^{p+\frac 32}}{2(p-k_m+\frac 12) w^{2p-2k_m+1}}
	\sum_{\ell=0}^{m-1} \frac{1}{\ell!} 
	  (-1) \widehat{dy} (-\widehat{\partial_y})^{\ell} \frac{1}{\widehat{dy} (dx)^{\ell}}   \sum_{\substack{I_1\sqcup \cdots \sqcup I_\ell = \set{m-1}\\ I_1,\dots,I_\ell \not=\emptyset}} \overline{\omega}(w)\times
	  \\ \notag & \qquad  \prod_{i=1}^\ell - (2(2p+1))^{1-|I_i|} u^{p+3/2}
	  \sum_{s_i=0}^{\infty}\frac{(-2)^{|I_i|-2}}{(s_i)!}(|I_i|-1)_{s_i}\left(-p+\sumk_{I_i}+s_i+\frac{1}{2}\right)_{|I_i|-2} \times
	  \\ \notag & \qquad \qquad (2p-2\sumk_{I_i}-2|I_i|-4s_i+3) \frac{dw}{w^{2p-2\sumk_{I_i}-2|I_i|-4s_i+4}}
\end{align}
(this and subsequent steps require some care for small $\ell$, but one can check by direct inspection that with our notation the final identity holds nevertheless). Note that for any function $F=F(w)$
\begin{align}
	& \Res_{w=0} \frac{u^{p+\frac 32}}{2(p-k_m+\frac 12) w^{2p-2k_m+1}}
	\sum_{\ell=0}^{m-1} \frac{1}{\ell!} 
	(-1) \widehat{dy} (-\widehat{\partial_y})^{\ell}F(w)
	\\ \notag &
	= \sum_{\ell=0}^{m-1}\frac{1}{\ell!}u^{p+\frac 32 - (p+\frac 12)(\ell-1)} (-1)^{\ell-1} \Big(\frac{2k_m}{2p+1} \Big)_{\ell-1} \Res_{w=0} w^{2k_m+(2p+1)(\ell-2)-1} F(w)dw.
\end{align}
Moreover,
\begin{align}
	\frac{(dw)^{\ell+1}}{\widehat{dy} (dx)^{\ell}} = u^{-(p+\frac 12) - \ell} \frac{(-1)^{\ell+1}}{(2p+1)2^\ell} \frac{w^{2p+2+3\ell}}{(1-w^4)^\ell}.
\end{align}
Thus,~\eqref{eq:combxcompdualside} is equal to 
\begin{align}
	& \Res_{w=0}
	\sum_{\ell=0}^{m-1} \frac{1}{\ell!} \frac{u^{p+\frac 32-(p+\frac 32)\ell}}{(2p+1)2^\ell}  \Big(\frac{2k_m}{2p+1} \Big)_{\ell-1} \sum_{s_0=0}^\infty \binom{s_0+\ell-1}{s_0} w^{2p+2+3\ell+2k_m+(2p+1)(\ell-2)-1+4s_0} \times 
	\\ \notag & \quad 
 \Bigg(\sum_{\substack{I_1\sqcup \cdots \sqcup I_\ell = [m-1]\\ I_1,\dots,I_\ell \not=\emptyset}} 
	\prod_{i=1}^\ell - (2(2p+1))^{1-|I_i|} u^{p+\frac 32}
	\sum_{s_i=0}^{\infty}\frac{(-2)^{|I_i|-2}}{(s_i)!}(|I_i|-1)_{s_i}\times 
	\\ \notag & \qquad \qquad \left(-p+\sumk_{I_i}+s_i+\frac{1}{2}\right)_{|I_i|-2} (2p-2\sumk_{I_i}-2|I_i|-4s_i+3) \frac{1}{w^{2p-2\sumk_{I_i}-2|I_i|-4s_i+4}}\Bigg) \overline{\omega}(w)
	\\ \notag &
	=   -(2(2p+1))^{1-m}u^{p+\frac 32} \Res_{w=0}\ \overline{\omega}(w) \cdot
	\sum_{\ell=0}^{m-1} \frac{1}{\ell!} (-(2p+1))^{\ell-1}\Big(\frac{2k_m}{2p+1} \Big)_{\ell-1} \sum_{s_0=0}^\infty \binom{s_0+\ell-1}{s_0} \times 
	\\ \notag & \quad 
	\Bigg(\sum_{\substack{I_1\sqcup \cdots \sqcup I_\ell = [m-1]\\ I_1,\dots,I_\ell \not=\emptyset}} 
	\prod_{i=1}^\ell 
	\sum_{s_i=0}^{\infty}\frac{(-2)^{|I_i|-2}(|I_i|-1)_{s_i}}{(s_i)!}\left(-p+\sumk_{I_i}+s_i+\frac{1}{2}\right)_{|I_i|-2}\times 
		\\ \notag & \qquad \qquad
	 (2p-2\sumk_{I_i}-2|I_i|-4s_i+3) \Bigg) 
	 \frac{ \overline{\omega}(w) }{w^{2p-2\sumk_{\set{m}}-2m+2-4(s_0+s_1+\cdots+s_m)+1}}.
\end{align}
By Lemma~\ref{lem:2ndCombLemma} given below, the latter expression is indeed equal to the right hand side of~\eqref{eq:residue-tree}.
\end{proof}

A direct corollary of this lemma is the following proposition:
\begin{proposition} \label{prop:xydual} For any $g\geq 0$, $n\geq 1$, $(g,n)\not=(0,1)$, $k_1,\dots,k_n = 1,\dots,p$, we have:
\begin{align} \label{eq:prop:xydual}
				A^{(g),\vee}_n(k_1,\dots,k_n) & = \frac{u^{-\sumk_{\set{n}}-n}}{2^n(2p+1)^n} \sum_{\ell=1}^n \frac{(-2)^\ell (2p+1)^\ell u^{(p+\frac 32)\ell}}{\ell!}  \sum_{\substack{I_1\sqcup \cdots \sqcup I_\ell = \set{n} \\ I_1,\dots,I_\ell \not=\emptyset}} \prod_{i=1}^\ell \Bigg(\Res_{w_i=0} \sum_{s_i=0}^{\infty} \frac{(|I_i|-1)_{s_i}}{s_i!} \times
	\\ \notag & \qquad \qquad
	(-2)^{|I_i|-2}\Big(-p+\sumk_{I_i}+s_i+\frac 12\Big)_{|I_i|-2}
	\frac{1}{w_i^{2p-2\sumk_{I_i} -2|I_i|-4s_i+3}} \Bigg) \overline{\omega}^{(g)}_\ell
	\\ \notag & \quad
	+ \delta_{g,0} \prod_{i=1}^n \Bigg( \Res_{w_i=0} \frac{u^{p-k_i+\frac 12}}{2(p-k_i+\frac 12) w^{2p-2k_i+1}}  \Bigg)
	\omega^{(0),\sim}_n,
\end{align}
where the second summand contributes only some terms constant in $u$ for $n\geq 3$ (and $g=0$). 
\end{proposition}

\begin{proof} By the argument presented above, it is a direct corollary of Lemma~\ref{lem:vee-to-veevee} combined Lemma~\ref{lem:veevee-tooperator} that allows to rewrite the right hand side of~\eqref{eq:simplification-of-graphs} as the first summand and Remark~\ref{rem:sim-explanation} that explains the second summand. 
\end{proof}

\subsection{Combinatorial lemma} The goal of this section is to prove the following lemma:

 \begin{lemma}\label{lem:2ndCombLemma} For any $m \geq 1$ we have the following identity of polynomials in $s$, $p$, $k_1,\dots,k_m$:
 	\begin{align} \label{eq:SecondCombIdentity}
 		&  
 		\frac{(m-1)_s}{s!} (-2)^{m-2} \Big(-p+\sumk_{\set{m}}+s+\frac 12\Big)_{m-2} 
 		\\ \notag &
 		= \sum_{\ell=0}^{m-1} \frac{1}{\ell!} \sum_{\substack{I_1\sqcup \cdots \sqcup I_\ell = \set{m-1}\\ I_1,\dots,I_\ell \not=\emptyset}} (-(2p+1))^{\ell-1} \Big(\frac {2k_m}{2p+1}\Big)_{\ell-1}
 		\sum_{\substack {s_0+s_1+\cdots+s_\ell = s \\ s_0,s_1,\dots,s_\ell \geq 0}} \binom{s_0+\ell-1}{s_0} \times
 		\\ \notag & \quad
 		\prod_{i=1}^\ell
 		\frac{(|I_i|-1)_{s_i}}{s_i!} (-2)^{|I_i|-2} \Big(-p+\sumk_{I_i}+s_i+\frac 12\Big)_{|I_i|-2} (2p-2\sumk_{I_i}-2|I_i|-4s_i+3).
 	\end{align}
 \end{lemma} 


\begin{remark}
For $m=1$ it is an identity of \emph{Laurent} polynomials. We include this case for completeness; it is also the only case that has a non-trivial contribution from the summand with $\ell=0$ on the right hand side. 
\end{remark}

\begin{proof} Firstly, divide~\eqref{eq:SecondCombIdentity} by $(-2)^{m-2}$, and then shift $p$ by $1/2$ (that is, new $p$ is old $p+1/2$). We have:
\begin{align} \label{eq:SecondIdentityBeforeGeneratingSeries}
	&  
	\frac{(m-1)_{s}}{s!} \Big(-p+k_{\set{m}}+s+1\Big)_{m-2} 
	\\ \notag &
	= \sum_{\ell=0}^{m-1} \frac{1}{\ell!} \sum_{\substack{I_1\sqcup \cdots \sqcup I_\ell = \set{m-1}\\ I_1,\dots,I_\ell \not=\emptyset}} p^{\ell-1} \Big(\frac {k_m}{p}\Big)_{\ell-1}
	 \sum_{\substack {s_0+s_1+\cdots+s_\ell = s \\ s_0,s_1,\dots,s_\ell \geq 0}} \binom{s_0+\ell-1}{s_0} \times
	\\ \notag & \quad
	\prod_{i=1}^\ell
	\frac{(|I_i|-1)_{s_i}}{s_i!} \Big(-p+k_{I_i}+s_i+1\Big)_{|I_i|-2} \Big(-p+k_{I_i}+s_i+|I_i|-1+s_i\Big)
\end{align}
Now we arrange~\eqref{eq:SecondIdentityBeforeGeneratingSeries} into a generating series as 
\begin{align} \label{eq:SecondIdentityGeneratingSeries} 
	& \sum_{s=0}^{\infty}  
	\frac{(m-1)_s}{s!}  \Big(-p+k_{\set{m}}+s+1\Big)_{m-2} z^{s}
= \sum_{\ell=0}^{m-1} \frac{1}{\ell!} \sum_{\substack{I_1\sqcup \cdots \sqcup I_\ell = \set{m-1}\\ I_1,\dots,I_\ell \not=\emptyset}} p^{\ell-1} \Big(\frac {k_m}{p}\Big)_{\ell-1}
(1-z)^{-\ell} \times
\\ \notag & \quad
\prod_{i=1}^\ell\sum_{s_i=0}^{\infty}
\frac{(|I_i|-1)_{s_i}}{s_i!} \Big(-p+k_{I_i}+s_i+1\Big)_{|I_i|-2} \Big(-p+k_{I_i}+s_i+|I_i|-1+s_i\Big)z^{s_i}.
\end{align}
The idea now is to expand everything in the Pochhammer symbols $(s+1)_{a-1}$, $a\geq 1$ and use that for $a\geq 0$
we have 
$a!(1-z)^{-(a+1)}= \sum_{s=0}^\infty (s+1)_{a} z^s$.
This makes both sides of~\eqref{eq:SecondIdentityGeneratingSeries} a polynomial in $X\coloneqq (1-z)^{-1}$. By a straightforward bruteforce calculation, \eqref{eq:SecondIdentityGeneratingSeries} can be rewritten as 
\begin{align}
	& \sum_{a=m-1}^{2m-3} X^a \frac{(a-1)!}{(a-m+1)! (2m-3-a)!} (-p+k_{\set{m}}+1)_{2m-3-a}
	\\ \notag &
	= \sum_{\ell=0}^{m-1} \frac{1}{\ell!} \sum_{\substack{I_1\sqcup \cdots \sqcup I_\ell = \set{m-1}\ I_1,\dots,I_\ell \not=\emptyset}} p^{\ell-1} \Big(\frac {k_m}{p}\Big)_{\ell-1} X^\ell
	\times
	\\ \notag & \qquad
	\prod_{i=1}^\ell \sum_{a_i=|I_i|-1}^{2|I_i|-2} X^{a_i} \frac{a_i!}{(a_i-|I_i|+1)! (2|I_i|-2-a_i)!}
	\Big(-p+k_{I_i}+1\Big)_{2|I_i|-2-a_i}.
\end{align}
We should note that in this equation the $k_i$ with $i\in\set{m-1}$ are shifted by $1$, that is, the new $k_i$ is the old $k_i-1$ if $i\in\set{m-1}$ and the new $k_m$ is the old $k_m$. Further dividing by both sides of the identity by $X^{m-1} = \prod_{i=1}^\ell X^{|I_i|}$ we get
\begin{align} \label{eq:Seconddentity-finalformbeforeexp}
	& \sum_{c=0}^{m-2} X^c \frac{(m-2+c)!}{c! (m-2-c)!} (-p+k_{\set{m}}+1)_{m-2-c}
	\\ \notag &
	= \sum_{\ell=0}^{m-1} \frac{1}{\ell!} \sum_{\substack{I_1\sqcup \cdots \sqcup I_\ell = \set{m-1}\\ I_1,\dots,I_\ell \not=\emptyset}} p^{\ell-1} \Big(\frac {k_m}{p}\Big)_{\ell-1}
	\prod_{i=1}^\ell \sum_{c_i=0}^{|I_i|-1} X^{c_i} \frac{(|I_i|-1+c_i)!}{c_i! (|I_i|-1-c_i)!}
	\Big(-p+k_{I_i}+1\Big)_{|I_i|-1-c_i}.
\end{align}

Now note that for any $n\geq 1$ we have the following auxiliary identity:
\begin{align}
\sum_{c=0}^{n-1} X^c \frac{(n-1+c)!}{c! (n-1-c)!} (x+y_{\set{n}})_{n-1-c} = \exp\Big(\sum_{d=1}^\infty \frac{X^d}{d} \sum_{i=1}^n \Delta_i^d \Big) (x+y_{\set{n}})_{n-1},
\end{align}
where $\Delta_i$ is the backwards difference operator in $i$th variable whose  action on a polynomial $P=P(y_{\set{n}})$ is defined as $\Delta_i \colon P\mapsto (1-e^{-\partial_{y_i}}) P$.
Applying this auxiliary identity for $n=m-1$, and $x=-p+k_m+1$, and $y_{\set{m-1}}=k_{\set{m-1}}$ on the left hand side of~\eqref{eq:Seconddentity-finalformbeforeexp}, and for $x=-p+1$, $n = |I_i|$ and $y_{\set{n}} = k_{I_i}$ to each of the $\ell$ factors on the right hand side of~\eqref{eq:Seconddentity-finalformbeforeexp}, we rewrite~\eqref{eq:Seconddentity-finalformbeforeexp} as
\begin{align}\label{eq:SecondCombIdentity-expform}
	& \exp\Big(\sum_{d=1}^\infty \frac{X^d}{d} \sum_{i=1}^n \Delta_i^d \Big)  (-p+k_{\set{m}}+1)_{m-2}
	\\ \notag &
	= \exp\Big(\sum_{d=1}^\infty \frac{X^d}{d} \sum_{i=1}^n \Delta_i^d \Big) \sum_{\ell=0}^{m-1} \frac{1}{\ell!} \sum_{\substack{I_1\sqcup \cdots \sqcup I_\ell = \set{m-1}\\ I_1,\dots,I_\ell \not=\emptyset}} p^{\ell-1} \Big(\frac {k_m}{p}\Big)_{\ell-1}
	\prod_{i=1}^\ell 	\Big(-p+k_{I_i}+1\Big)_{|I_i|-1}.
\end{align}
Thus, in order to prove~\eqref{eq:SecondCombIdentity-expform} (whch is equivalent to~\eqref{eq:SecondCombIdentity}), we only have to prove that 
\begin{align} \label{eq:identity-final-form}
	& (-p+k_{\set{m}}+1)_{m-2}
	=  \sum_{\ell=0}^{m-1} \frac{1}{\ell!} \sum_{\substack{I_1\sqcup \cdots \sqcup I_\ell = \set{m-1}\\ I_1,\dots,I_\ell \not=\emptyset}} p^{\ell-1} \Big(\frac {k_m}{p}\Big)_{\ell-1}
	\prod_{i=1}^\ell 	\Big(-p+k_{I_i}+1\Big)_{|I_i|-1}.
\end{align}
This can be done by induction on $m$. The case $m=1$ is a bit exceptional since it is an equality of two Laurent polynomials that reads $(-p+k_1)^{-1} = (-p+k_1)^{-1}$. For $m=2$ we get the identity $1=1$. For any $q\geq 3$, one can check (using as the induction assumption that~\eqref{eq:identity-final-form} holds for $m=q-1$) that~\eqref{eq:identity-final-form} also holds for $m=q$ at $k_1=0$. By symmetry in the arguments $k_1,\dots,k_{m-1}$ we conclude that~\eqref{eq:identity-final-form} also holds for $m=q$ at $k_i=0$ for all $i=1,\dots,m-1$. Since the only polynomial of degree $m-2$ in the variables $k_{\set{m-1}}$ that vanished at $k_i=0$ for all $i=1,\dots,m-1$ is contant zero, we conclude that~\eqref{eq:identity-final-form} holds for $m=q$ and thus, by induction, for any $m\geq 1$. 
\end{proof}

\printbibliography

\end{document}